\documentclass[runningheads,a4paper]{llncs}

\usepackage{amssymb,amsmath,tikz}
\usepackage{graphicx}

\usepackage{url}
\usepackage{enumerate}
\urldef{\maila}\path|cel@mit.edu|
\newcommand{\keywords}[1]{\par\addvspace\baselineskip
\noindent\keywordname\enspace\ignorespaces#1}

\makeatletter
\newcommand{\text@hyphens}{\mathcode`\-=`\-\relax}
\newcommand{\id}[1]{\ensuremath{\mathit{\text@hyphens#1}}}
\makeatother

\begin{document}

\mainmatter  

\title{Upper Bounds on Number of Steals in Rooted Trees}

\titlerunning{Upper Bounds on Number of Steals in Rooted Trees}

%
%
\author{Charles E. Leiserson\inst{1}\and Tao B. Schardl\inst{1}\and Warut Suksompong\inst{2}%
}
\authorrunning{C. E. Leiserson, T. B. Schardl, W. Suksompong}

\institute{MIT Computer Science and Artificial Intelligence Laboratory\\
32 Vassar St, Cambridge, MA 02139, USA\\
\email{\{cel,neboat\}@mit.edu}
\and Department of Computer Science, Stanford University\\
353 Serra Mall, Stanford, CA 94305, USA\\
\email{warut@cs.stanford.edu}
}

%
%

\maketitle

\begin{abstract}
Inspired by applications in parallel computing, we analyze the setting of work stealing in multithreaded computations. We obtain tight upper bounds on the number of steals when the computation can be modeled by rooted trees. In particular, we show that if the computation with $n$ processors starts with one processor having a complete $k$-ary tree of height $h$ (and the remaining $n-1$ processors having nothing), the maximum possible number of steals is $\sum_{i=1}^n(k-1)^i\binom{h}{i}$.
\end{abstract}

\section{Introduction}

The two main scheduling paradigms that are commonly used for scheduling multithreaded computations are \textit{work sharing} and \textit{work stealing}. The two paradigms differ in how the threads are distributed. In work sharing, the scheduler migrates new threads to other processors so that underutilized processors have more work to do. In work stealing, on the other hand, underutilized processors attempt to ``steal'' threads from other processors. Intuitively, the migration of threads occurs more often in work sharing than in work stealing, since no thread migration occurs in work stealing when all processors have work to do, whereas a work-sharing scheduler always migrates threads regardless of the current utilization of the processors.

The idea of work stealing has been around at least as far back as the 1980s. Burton and Sleep's work \cite{BurtonSl81} on functional programs on a virtual tree of processors and Halstead's work \cite{Halstead84} on the implementation of Multilisp are among the first to outline the idea of work stealing. These authors point out the heuristic benefits of the work-stealing paradigm with regard to space and communication. Since then, several other researchers \cite{AgrawalHeLe06,AroraBlPl98,DinanLaSaKrNi09,KarpZh93,SuksompongLeSc16} have found applications of the work-stealing paradigm or analyzed the paradigm in new ways.

An important contribution to the literature of work stealing was made by Blumofe and Leiserson \cite{BlumofeLe99}, who gave the first provably good work-stealing scheduler for multithreaded computations with dependencies. Their scheduler executes a fully strict (i.e., well-structured) multithreaded computations on $P$ processors within an expected time of $T_1/P+O(T_\infty)$, where $T_1$ is the minimum serial execution time of the multithreaded computation (the \textit{work} of the computation) and $T_\infty$ is the minimum execution time with an infinite number of processors (the \textit{span} of the computation.) Furthermore, the scheduler has provably good bounds on total space and total communication in any execution.

This paper analyzes upper bounds on the number of steals in multithreaded computations. While the existing literature has dealt extensively with probabilistic and average-case analysis, it has not covered worst-case analysis on work stealing. We obtain tight upper bounds on the number of steals when the computation can be modeled by rooted trees. In particular, we show that if the computation with $n$ processors starts with one processor having a complete $k$-ary tree of height $h$ (and the remaining $n-1$ processors having nothing), the maximum possible number of steals is $\sum_{i=1}^n(k-1)^i\binom{h}{i}$.

The remainder of this paper is organized as follows. Section \ref{sec:setting} introduces the setting that we will analyze throughout the paper. Section \ref{sec:notation} introduces some notation that we will use later in the paper. Section \ref{sec:onestartingproc} analyzes the case where at the beginning all the work is owned by one processor, and the computation tree is a binary tree. Section \ref{sec:karytrees} generalizes Section \ref{sec:onestartingproc} to the case where the computation tree is an arbitrary rooted tree. Section \ref{sec:manystartingprocs} generalizes Section \ref{sec:karytrees} one step further and analyzes the case where at the beginning all processors could possibly own work. Finally, Section \ref{sec:conclusion} concludes and suggests directions for future work.

\section{Setting}
\label{sec:setting}
This section introduces the setting that we will analyze throughout the paper. We describe the setting while keeping the mathematical core as our focus. Suksompong \cite{Suksompong14} provides more detail on the connection between this setting and the work-stealing paradigm.

Suppose that there are $P$ processors. Each processor owns some work, which we model as a computation tree. Throughout this paper, we assume that computation trees are rooted trees with no ``singleton nodes'' (i.e., nodes that have exactly one child.) At any point in the execution of work stealing, a processor is allowed to perform the following two-step operation:

\begin{enumerate}
\item Finish its own work, i.e., destroy its computation tree.
\item Pick another processor that currently owns a computation tree with more than one node, and ``steal'' from that processor.
\end{enumerate}

We now describe how a steal proceeds. Suppose that processor $P_2$ steals from processor $P_1$, and assume that the root node of $P_1$ has $m$ children. If $m>2$, then $P_2$ steals the rightmost subtree, leaving the root node and the other $m-1$ subtrees of $P_1$ intact. If $m=2$, on the other hand, then $P_2$ steals the right subtree, leaving $P_1$ with the left subtree, while the root node disappears. 

An example of a steal is shown in Figure \ref{fig:karyexample}. In this example, the root node has four children, and therefore the steal takes away the rightmost subtree and leaves the remaining three subtrees intact. This definition of stealing in rooted trees is not arbitrary. It is, in fact, the one commonly used in parallel computing.

\begin{figure}
\centering
\begin{tabular}{@{}c@{}}
\begin{tikzpicture}[scale=1]

  \path (0.5,2) coordinate (P1) node[left=0.05cm] {2};
  \path (2,2) coordinate (P2) node[right=0.05cm] {3};
  \path (3.5,2) coordinate (P3) node[right=0.05cm] {4};
  \path (5,2) coordinate (P4) node[right=0.05cm] {5};
  \path (2.75,3) coordinate (P5) node[right=0.09cm] {1};
  \path (4.5,1) coordinate (P6) node[right=0.02cm] {11};
  \path (5.5,1) coordinate (P7) node[right=0.02cm] {12};
  \path (5,0) coordinate (P8) node[right=0.02cm] {13};
  \path (6,0) coordinate (P9) node[right=0.02cm] {14};
  \path (1.5,1) coordinate (P10) node[right=0.02cm] {6};
  \path (2,1) coordinate (P11) node[right=0.02cm] {7};
  \path (2.5,1) coordinate (P12) node[right=0.02cm] {8};
  \path (3,1) coordinate (P13) node[right=0.02cm] {9};
  \path (4,1) coordinate (P14) node[left=0.02cm] {10};
  \foreach \i in {1,...,14}
  {
    \fill (P\i) circle (2pt);
  }
  \draw (P1)--(P5) (P2)--(P5) (P3)--(P5) (P4)--(P5) (P4)--(P6) (P4)--(P7) (P7)--(P8) (P7)--(P9) (P2)--(P10) (P2)--(P11) (P2)--(P12) (P3)--(P13) (P3)--(P14);

  \path (8,1) coordinate (Q1) node[right=0.02cm] {19};
  \path (9,1) coordinate (Q2) node[right=0.02cm] {20};
  \path (7.5,2) coordinate (Q3) node[right=0.02cm] {16};
  \path (8.5,2) coordinate (Q4) node[right=0.02cm] {17};
  \path (9.5,2) coordinate (Q5) node[right=0.02cm] {18};
  \path (8.5,3) coordinate (Q6) node[right=0.02cm] {15};
  \foreach \i in {1,...,6}
  {
    \fill (Q\i) circle (2pt);
  }
  \draw (Q1)--(Q4) (Q2)--(Q4) (Q3)--(Q6) (Q4)--(Q6) (Q5)--(Q6);

  \path (2.75,3.25) coordinate node[above=0.1cm] {$P_1$};
  \path (8.5,3.25) coordinate node[above=0.1cm] {$P_2$};

\end{tikzpicture}
\\[\abovecaptionskip]
\small (a) Initial configuration
\end{tabular}

\vspace{5mm}

\begin{tabular}{@{}c@{}}
\begin{tikzpicture}[scale=1]

  \path (0.5,1) coordinate (P1) node[left=0.05cm] {2};
  \path (2,1) coordinate (P2) node[right=0.05cm] {3};
  \path (3.5,1) coordinate (P3) node[right=0.05cm] {4};
  \path (2,2) coordinate (P4) node[right=0.09cm] {1};
  \path (1.5,0) coordinate (P5) node[right=0.02cm] {6};
  \path (2,0) coordinate (P6) node[right=0.02cm] {7};
  \path (2.5,0) coordinate (P7) node[right=0.02cm] {8};
  \path (3,0) coordinate (P8) node[right=0.02cm] {9};
  \path (4,0) coordinate (P9) node[left=0.02cm] {10};
  \foreach \i in {1,...,9}
  {
    \fill (P\i) circle (2pt);
  }
  \draw (P1)--(P4) (P2)--(P4) (P3)--(P4) (P2)--(P5) (P2)--(P6) (P2)--(P7) (P3)--(P8) (P3)--(P9);

  \path (5.5,1) coordinate (Q1) node[right=0.02cm] {11};
  \path (6.5,1) coordinate (Q2) node[right=0.02cm] {12};
  \path (6,2) coordinate (Q3) node[right=0.05cm] {5};
  \path (6,0) coordinate (Q4) node[right=0.02cm] {13};
  \path (7,0) coordinate (Q5) node[right=0.02cm] {14};
  \foreach \i in {1,...,5}
  {
    \fill (Q\i) circle (2pt);
  }
  \draw (Q1)--(Q3) (Q2)--(Q3) (Q2)--(Q4) (Q2)--(Q5);

  \path (2,2.25) coordinate node[above=0.1cm] {$P_1$};
  \path (6,2.25) coordinate node[above=0.1cm] {$P_2$};

\end{tikzpicture} 
\\[\abovecaptionskip]
\small (b) Processor $P_2$ steals from processor $P_1$
\end{tabular}
\caption{Example of a steal} \label{fig:karyexample}

\end{figure}

The question that we analyze in this paper is as follows: Given any starting configuration of computation trees corresponding to the processors, what is the maximum possible number of steals that can occur in an execution of work stealing?

Before we begin our analysis on the number of steals, we consider an example of a complete execution of work stealing in Figure \ref{fig:worstcaseexample}. There are two processors in this example, $P_1$ and $P_2$, and the initial computation trees corresponding to the two processors are shown in Figure \ref{fig:worstcaseexample}(a). Three steals are performed, yielding the trees in Figures \ref{fig:worstcaseexample}(b), \ref{fig:worstcaseexample}(c), and \ref{fig:worstcaseexample}(d) respectively. In this example, one can check that the maximum number of steals that can be performed in the execution is also 3. 
\begin{figure}
\centering

\begin{tabular}{@{}c@{}}
\begin{tikzpicture}[scale=1]

  \path (0,1) coordinate (P1) node[right=0.02cm] {4};
  \path (1,1) coordinate (P2) node[right=0.02cm] {5};
  \path (2,1) coordinate (P3) node[right=0.02cm] {6};
  \path (3,1) coordinate (P4) node[right=0.02cm] {7};
  \path (0.5,2) coordinate (P5) node[right=0.02cm] {2};
  \path (2.5,2) coordinate (P6) node[right=0.02cm] {3};
  \path (1.5,3) coordinate (P7) node[right=0.02cm] {1};
  \path (-0.5,0) coordinate (P8) node[right=0.02cm] {8};
  \path (0.5,0) coordinate (P9) node[right=0.02cm] {9};
  \foreach \i in {1,...,9}
  {
    \fill (P\i) circle (2pt);
  }
  \draw (P1)--(P5) (P2)--(P5) (P3)--(P6) (P4)--(P6) (P5)--(P7) (P6)--(P7) (P1)--(P8) (P1)--(P9);

  \path (5,1) coordinate (Q1) node[right=0.02cm] {13};
  \path (6,1) coordinate (Q2) node[right=0.02cm] {14};
  \path (7,1) coordinate (Q3) node[right=0.02cm] {15};
  \path (8,1) coordinate (Q4) node[right=0.02cm] {16};
  \path (5.5,2) coordinate (Q5) node[right=0.02cm] {11};
  \path (7.5,2) coordinate (Q6) node[right=0.02cm] {12};
  \path (6.5,3) coordinate (Q7) node[right=0.02cm] {10};
  \foreach \i in {1,...,7}
  {
    \fill (Q\i) circle (2pt);
  }
  \draw (Q1)--(Q5) (Q2)--(Q5) (Q3)--(Q6) (Q4)--(Q6) (Q5)--(Q7) (Q6)--(Q7);

  \path (1.5,3.25) coordinate node[above=0.1cm] {$P_1$};
  \path (6.5,3.25) coordinate node[above=0.1cm] {$P_2$};

\end{tikzpicture}
\\[\abovecaptionskip]
\small (a) Initial configuration
\end{tabular}

\vspace{5mm}

\begin{tabular}{@{}c@{}}
\begin{tikzpicture}[scale=1]

  \path (0,1) coordinate (P1) node[right=0.02cm] {4};
  \path (1,1) coordinate (P2) node[right=0.02cm] {5};
  \path (0.5,2) coordinate (P3) node[right=0.02cm] {2};
  \path (-0.5,0) coordinate (P4) node[right=0.02cm] {8};
  \path (0.5,0) coordinate (P5) node[right=0.02cm] {9};
  \foreach \i in {1,...,5}
  {
    \fill (P\i) circle (2pt);
  }
  \draw (P1)--(P3) (P2)--(P3) (P1)--(P4) (P1)--(P5);

  \path (3,1) coordinate (Q1) node[right=0.02cm] {6};
  \path (4,1) coordinate (Q2) node[right=0.02cm] {7};
  \path (3.5,2) coordinate (Q3) node[right=0.02cm] {3};
  \foreach \i in {1,...,3}
  {
    \fill (Q\i) circle (2pt);
  }
  \draw (Q1)--(Q3) (Q2)--(Q3);

  \path (0.5,2.25) coordinate node[above=0.1cm] {$P_1$};
  \path (3.5,2.25) coordinate node[above=0.1cm] {$P_2$};

\end{tikzpicture}
\\[\abovecaptionskip]
\small (b) Processor $P_2$ steals from processor $P_1$
\end{tabular}

\vspace{5mm}

\begin{tabular}{@{}c@{}}
\begin{tikzpicture}[scale=1]

  \path (0,1) coordinate (P1) node[right=0.02cm] {4};
  \path (-0.5,0) coordinate (P2) node[right=0.02cm] {8};
  \path (0.5,0) coordinate (P3) node[right=0.02cm] {9};
  \foreach \i in {1,...,3}
  {
    \fill (P\i) circle (2pt);
  }
  \draw (P1)--(P2) (P1)--(P3);

  \path (3,1) coordinate (Q1) node[right=0.02cm] {5};
  \foreach \i in {1,...,1}
  {
    \fill (Q\i) circle (2pt);
  }

  \path (0,1.25) coordinate node[above=0.1cm] {$P_1$};
  \path (3,1.25) coordinate node[above=0.1cm] {$P_2$};

\end{tikzpicture}
\\[\abovecaptionskip]
\small (c) Processor $P_2$ steals from processor $P_1$
\end{tabular}

\vspace{5mm}

\begin{tabular}{@{}c@{}}
\begin{tikzpicture}[scale=1]

  \path (0,1) coordinate (P1);
  \foreach \i in {1,...,1}
  {
    \fill (P\i) circle (2pt) node[right=0.02cm] {8};
  }

  \path (2,1) coordinate (Q1);
  \foreach \i in {1,...,1}
  {
    \fill (Q\i) circle (2pt) node[right=0.02cm] {9};
  }

  \path (0,1.25) coordinate node[above=0.1cm] {$P_1$};
  \path (2,1.25) coordinate node[above=0.1cm] {$P_2$};

\end{tikzpicture}
\\[\abovecaptionskip]
\small (d) Processor $P_2$ steals from processor $P_1$
\end{tabular}

\caption{Example of an execution of work stealing} \label{fig:worstcaseexample}

\end{figure}

\section{Notation}
\label{sec:notation}
This section introduces some notation that we will use later in the paper.

\subsection*{Binomial Coefficients}

For positive integers $n$ and $k$, the binomial coefficient $\binom{n}{k}$ is defined as  
\[\dbinom{n}{k}=\begin{cases} \dfrac{n!}{k!(n-k)!} &\mbox{if } n\geq k\\ 
                              0 & \mbox{if } n<k. \end{cases}\] 
The binomial coefficients satisfy Pascal's identity \cite{CormenLeRi09} 
\begin{equation} \label{eq:pascal}
\binom{n}{k}=\binom{n-1}{k}+\binom{n-1}{k-1} 
\end{equation}
for all integers $n,k\geq 1$. One can verify the identity directly using the definition.

\subsection*{Trees}
Let $\id{EMPT}$ denote the empty tree with no nodes, and let $\id{TRIVT}$ denote the trivial tree with only one node.

For $h\geq 0$, let $\id{CBT}(h)$ denote the complete binary tree with height $h$. For instance, the tree $\id{CBT}(4)$ is shown in Figure \ref{fig:CBT4}.

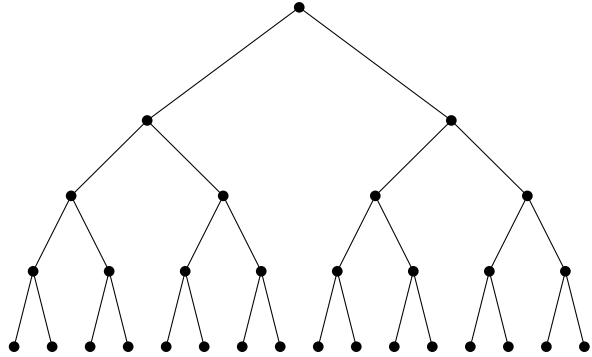
\begin{figure}
\begin{center}
\begin{tikzpicture}[scale=1]

  \foreach \i in {0,...,15}
  {
    \path (\i*0.5,0) coordinate (P\i);
    \fill (P\i) circle (2pt);
  }
  \foreach \i in {16,...,23}
  {
    \path (\i-15.75,1) coordinate (P\i);
    \fill (P\i) circle (2pt);
  }
  \foreach \i in {24,...,27}
  {
    \path (\i*2-47.25,2) coordinate (P\i);
    \fill (P\i) circle (2pt);
  }
  \foreach \i in {28,...,29}
  {
    \path (\i*4-110.25,3) coordinate (P\i);
    \fill (P\i) circle (2pt);
  }
  \foreach \i in {30,...,30}
  {
    \path (\i*8-236.25,4.5) coordinate (P\i);
    \fill (P\i) circle (2pt);
  }
  \draw (P0)--(P16) (P1)--(P16) (P2)--(P17) (P3)--(P17) (P4)--(P18) (P5)--(P18) (P6)--(P19) (P7)--(P19) (P8)--(P20) (P9)--(P20) (P10)--(P21) (P11)--(P21) (P12)--(P22) (P13)--(P22) (P14)--(P23) (P15)--(P23) (P16)--(P24) (P17)--(P24) (P18)--(P25) (P19)--(P25) (P20)--(P26) (P21)--(P26) (P22)--(P27) (P23)--(P27) (P24)--(P28) (P25)--(P28) (P26)--(P29) (P27)--(P29) (P28)--(P30) (P29)--(P30);

\end{tikzpicture}
\caption{The complete binary tree $\id{CBT}(4)$ of height 4} \label{fig:CBT4}
\end{center}
\end{figure}

For $k\geq 2, h\geq 0$, and $1\leq b\leq k-1$, let $\id{ACT}(b,k,h)$ denote the ``almost complete'' $k$-ary tree with $b\cdot k^h$ leaves. In particular, the root has $b$ children if $b\neq 1$ and $k$ children if $b=1$, and every other node has $0$ or $k$ children. For instance, the tree $\id{ACT}(2,3,2)$ is shown in Figure \ref{fig:almostperfectexample}. By definition, we have $\id{CBT}(h)=\id{ACT}(1,2,h)$. Also, the tree $\id{ACT}(b,k,h)$ is complete exactly when $b=1$.

\begin{figure}

\centering
\begin{tikzpicture}[scale=1]

  \path (0.5,1) coordinate (P1);
  \path (2,1) coordinate (P2);
  \path (3.5,1) coordinate (P3);
  \path (5.5,1) coordinate (P4);
  \path (7,1) coordinate (P5);
  \path (8.5,1) coordinate (P6);
  \path (2,2) coordinate (P7);
  \path (7,2) coordinate (P8);
  \path (4.5,3.5) coordinate (P9);
  \path (0,0) coordinate (P10);
  \path (0.5,0) coordinate (P11);
  \path (1,0) coordinate (P12);
  \path (1.5,0) coordinate (P13);
  \path (2,0) coordinate (P14);
  \path (2.5,0) coordinate (P15);
  \path (3,0) coordinate (P16);
  \path (3.5,0) coordinate (P17);
  \path (4,0) coordinate (P18);
  \path (5,0) coordinate (P19);
  \path (5.5,0) coordinate (P20);
  \path (6,0) coordinate (P21);
  \path (6.5,0) coordinate (P22);
  \path (7,0) coordinate (P23);
  \path (7.5,0) coordinate (P24);
  \path (8,0) coordinate (P25);
  \path (8.5,0) coordinate (P26);
  \path (9,0) coordinate (P27);
  \foreach \i in {1,...,27}
  {
    \fill (P\i) circle (2pt);
  }
  \draw (P1)--(P7) (P2)--(P7) (P3)--(P7) (P4)--(P8) (P5)--(P8) (P6)--(P8) (P7)--(P9) (P8)--(P9) (P1)--(P10) (P1)--(P11) (P1)--(P12) (P2)--(P13) (P2)--(P14) (P2)--(P15) (P3)--(P16) (P3)--(P17) (P3)--(P18) (P4)--(P19) (P4)--(P20) (P4)--(P21) (P5)--(P22) (P5)--(P23) (P5)--(P24) (P6)--(P25) (P6)--(P26) (P6)--(P27);

\end{tikzpicture}

\caption{The almost complete ternary tree $\id{ACT}(2,3,2)$ of height $3$ and root branching factor $2$} \label{fig:almostperfectexample}

\end{figure}
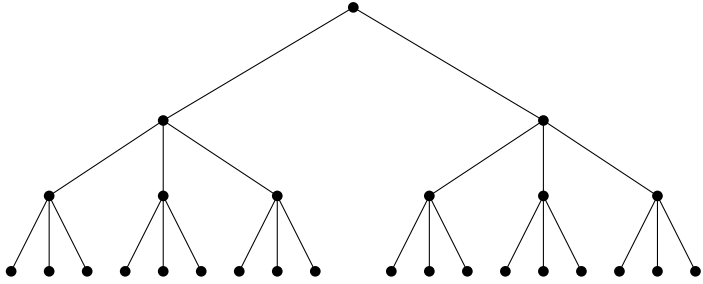

For any tree $T$, denote by $|T|$ the size of the tree $T$, i.e., the number of nodes in $T$.

\section{Binary Tree and One Starting Processor}
\label{sec:onestartingproc}

In this section, we establish the upper bound on the number of steals in the case where at the beginning, all the work is owned by one processor, and the computation tree is a binary tree. We define a potential function to help with establishing the upper bound, and we derive a recurrence relation for the potential function. The recurrence directly yields an algorithm that computes the upper bound in time $O(|T|n)$. Finally, we show that the maximum number of steals that can occur if our configuration starts with the tree $\id{CBT}(h)$ is $\sum_{i=1}^n\binom{h}{i}$.

\subsection*{Recurrence}

Suppose that we are given a configuration in which all processors but one start with an empty tree, while the exceptional processor starts with a tree $T$. How might a sequence of steals proceed? The first steal is fixed---it must split the tree $T$ into its left and right subtrees, $T_l$ and $T_r$. From there, one way to proceed is, in some sense, to be greedy. We obtain as many steals out of $T_l$ as we can, while keeping $T_r$ intact. As such, we have $P-1$ processors that we can use to perform steals on $T_l$, since the last processor must maintain $T_r$. Then, after we are done with $T_l$, we can perform steals on $T_r$ using all of our $P$ processors. This motivates the following definition.

\begin{definition}
Let $n\geq 0$ be an integer and $T$ a binary tree. The $n$th potential of $T$ is defined as the maximum number of steals that can be obtained from a configuration of $n+1$ processors, one of which has the tree $T$ and the remaining $n$ of which have empty trees. The $n$th potential of $T$ is denoted by $\Phi(T,n)$.
\end{definition}

If we only have one processor to work with, we cannot perform any steals, hence $\Phi(T,0)=0$ for any tree $T$. Moreover, the empty tree $\id{EMPT}$ and the trivial tree with a single node $\id{TRIVT}$ cannot generate any steals, hence $\Phi(\id{EMPT},n)=\Phi(\id{TRIVT},n)=0$ for all $n\geq 0$. 

In addition, the discussion leading up to the definition shows that if a binary tree $T$ has left subtree $T_l$ and right subtree $T_r$, then for any $n\geq 1$, we have \[\Phi(T,n)\geq 1+\Phi(T_l,n-1)+\Phi(T_r,n).\] By symmetry, we also have \[\Phi(T,n)\geq 1+\Phi(T_r,n-1)+\Phi(T_l,n).\] Combining the two inequalities yields
\[ \Phi(T,n)\geq 1+\max\{\Phi(T_l,n-1)+\Phi(T_r,n),\Phi(T_r,n-1)+\Phi(T_l,n)\}. \]

In the next theorem, we show that this inequality is in fact always an equality.

\begin{theorem}
\label{thm:worstcaserecurrence}
Let $T$ be a binary tree with at least 2 nodes, and let $T_l$ and $T_r$ be its left and right subtrees. Then for any $n\geq 1$, we have
\[ \Phi(T,n)= 1+\max\{\Phi(T_l,n-1)+\Phi(T_r,n),\Phi(T_r,n-1)+\Phi(T_l,n)\}. \]
\end{theorem}

\begin{proof}
Since we have already shown that the left-hand side is no less than the right-hand side, it only remains to show the reverse inequality.

Suppose that we are given any sequence of steals performed on $T$ using $n$ processors. As we have noted before, the first steal is fixed---it must split the tree $T$ into its two subtrees, $T_l$ and $T_r$. Each of the subsequent steals is performed either on a subtree of $T_l$ or a subtree of $T_r$ (not necessarily \textit{the} left or right subtrees of $T_l$ or $T_r$.) Assume for now that the last steal is performed on a subtree of $T_r$. That means that at any particular point in the stealing sequence, subtrees of $T_l$ occupy at most $n-1$ processors. (Subtrees of $T_l$ may have occupied all $n$ processors at different points in the stealing sequence, but that does not matter.) We can canonicalize the sequence of steals in such a way that the steals on subtrees of $T_l$ are performed first using $n-1$ processors, and then the steals on subtrees of $T_r$ are performed using $n$ processors. Therefore, in this case the total number of steals is at most $1+\Phi(T_l,n-1)+\Phi(T_r,n)$.

Similarly, if the last steal is performed on a subtree of $T_l$, then the total number of steals is at most $1+\Phi(T_r,n-1)+\Phi(T_l,n)$. Combining the two cases, we have \[ \Phi(T,n)\leq 1+\max\{\Phi(T_l,n-1)+\Phi(T_r,n),\Phi(T_r,n-1)+\Phi(T_l,n)\}, \] which gives us the desired equality.

\end{proof}

\subsection*{Algorithm}

When combined with the base cases previously discussed, Theorem \ref{thm:worstcaserecurrence} gives us a recurrence that we can use to compute $\Phi(T,n)$ for any binary tree $T$ and any value of $n$. But how fast can we compute the potential? The next corollary addresses that question. Recall from Section \ref{sec:notation} that $|T|$ denotes the size of the tree $T$.

\begin{corollary}
\label{cor:binaryalgo}
Let $T$ be a binary tree. There exists an algorithm that computes the potential $\Phi(T,n)$ in time $O(|T|n)$ and space $O(hn)$, where $h$ denotes the height of $T$.
\end{corollary}

\begin{proof}
The algorithm uses dynamic programming to compute $\Phi(T,n)$. For each subtree $T'$ of $T$ and each value $0\leq i\leq n$, the algorithm computes $\Phi(T',i)$ using the recurrence given in Theorem \ref{thm:worstcaserecurrence}. There are $O(|T|n)$ subproblems to solve, and each subproblem takes $O(1)$ time to compute. Hence the running time is $O(|T|n)$.

To optimize space, we can traverse the tree using post-order traversal. Whenever we have computed the values $\Phi(T',i)$ for two siblings, we use those values to compute $\Phi(T',i)$ for the parent and subsequently delete the values for the siblings. One can check that at any stage, if we do not consider the tree whose root node we are currently traversing, then the trees $T'$ for which we store the values $\Phi(T',i)$ have root nodes that are of pairwise different depth from the root node of $T$. Hence we store at most $hn$ values at any stage, and the algorithm takes space $O(hn)$.
\end{proof}

\subsection*{Complete Binary Trees}

An interesting special case is when the initial tree $T$ is a complete binary tree, i.e., a full binary tree in which all leaves have the same depth and every parent has two children. Recall from Section \ref{sec:notation} that $\id{CBT}(h)$ denotes the complete binary tree with height $h$. The next corollary establishes the potential of $\id{CBT}(h)$. Recall also from Section \ref{sec:notation} that $\binom{a}{b}=0$ if $a<b$.

\begin{corollary}
\label{cor:binarytrees}
We have
\begin{align}
\Phi(\id{CBT}(h),n)  &= \binom{h}{1}+\binom{h}{2}+\cdots+\binom{h}{n} \notag \\
                              &= \sum_{i=1}^{n}\binom{h}{i}, \notag 
\end{align}
\end{corollary}

\begin{proof}
The case $h=0$ holds, since 
\[\Phi(\id{CBT}(0),n)=\displaystyle\sum_{i=1}^{n}\dbinom{0}{i}=0.\] 
The case $n=0$ holds similarly. Now suppose that $h,n>0$. Since the two subtrees of $\id{CBT}(h)$ are symmetric, the recurrence in Theorem \ref{thm:worstcaserecurrence} yields

\[ \Phi(\id{CBT}(h),n) = 1 + \Phi(\id{CBT}(h-1),n-1) + \Phi(\id{CBT}(h-1),n).\]

We proceed by induction. Suppose that the formula for the potential values hold for $\id{CBT}(h-1)$. Using Pascal's identity (Equation \ref{eq:pascal}), we have
\begin{align}
\Phi(\id{CBT}(h),n)  &= 1 + \Phi(\id{CBT}(h-1),n-1) + \Phi(\id{CBT}(h-1),n) \notag \\
                              &= \binom{h-1}{0} + \sum_{i=1}^{n-1}\binom{h-1}{i} + \sum_{i=1}^n\binom{h-1}{i} \notag \\
                              &= \sum_{i=1}^{n}\left(\binom{h-1}{i-1}+\binom{h-1}{i}\right) \notag \\
                              &= \sum_{i=1}^n\binom{h}{i}, \notag
\end{align}
as desired.
\end{proof}

We can write the recurrence
\[ \Phi(\id{CBT}(h),n) = 1 + \Phi(\id{CBT}(h-1),n-1) + \Phi(\id{CBT}(h-1),n)\]
in the form
\[ F(h,n) = F(h-1,n-1) + F(h-1,n),\]
where $F(h,n) = \Phi(\id{CBT}(h),n)+1$. The recurrence equation of $F$ is the same as that of the binomial coefficients. On the other hand, the initial conditions of the two recurrences differ slightly. The initial conditions of the binomial coefficients are $\binom{h}{0}=1$ for all $h\geq 0$ and $\binom{0}{n}=0$ for all $n\geq 1$, while the initial conditions of $F$ are $F(h,0)=F(0,n)=1$ for all $h,n\geq 0$.

For fixed $n$, $\Phi(\id{CBT}(h),n)$ grows as $O(h^n)$. Indeed, one can obtain the (loose) bound 
\[\sum_{i=1}^n\dbinom{h}{i}\leq h^n\]
for $h,n\geq 2$, for example using the simple bound 
\[\dbinom{h}{i}=\dfrac{h(h-1)\cdots(h-i+1)}{i!}\leq h(h-1)^{i-1}\] 
and then computing a geometric sum.

When $n\geq h$, Corollary \ref{cor:binarytrees} implies that $\Phi(\id{CBT}(h),n)=\sum_{i=1}^{n}\binom{h}{i}=2^h-1$. The quantity $2^h-1$ is 1 less than the number of leaves in a complete binary tree of height $h$. Hence, the bound is equivalent to the trivial upper bound that the number of steals is at most the number of tasks in the tree. It is worth noting that the condition $n\geq h$ might hold in many practical applications, and we would be stuck with this upper bound in such situations. For instance, if we process less than 1 billion tasks using at least 30 processors, the condition $n\geq h$ is fulfilled.

We have established the upper bound on the number of steals in the case where at the beginning, all the work is owned by one processor, and the computation tree is a binary tree. In the next sections we will generalize to configurations where the work can be spread out at the beginning and take the form of arbitrary rooted trees as well.

\section{Rooted Tree and One Starting Processor}
\label{sec:karytrees}

In this section, we consider a generalization of Section \ref{sec:onestartingproc} to the configuration in which the starting tree is an arbitrary rooted tree. The key observation is that we can transform an arbitrary rooted tree into a ``left-child right-sibling'' binary tree that is equivalent to the rooted tree with respect to steals. With this transformation, an algorithm that computes the upper bound in time $O(|T|n)$ follows. Finally, we show that the maximum number of steals that can occur if our configuration starts with the tree $\id{ACT}(b,k,h)$ is $\sum_{i=1}^n(k-1)^i\binom{h}{i}+(b-1)\sum_{i=0}^{n-1}(k-1)^i\binom{h}{i}$.

\subsection*{Algorithm}

The key to establishing the upper bound in the case of arbitrary rooted trees is to observe that the problem can in fact be reduced to the case of binary trees, which we have already settled in Section \ref{sec:onestartingproc}. We transform any rooted tree into a binary tree by transforming each node with more than two children into a left-child right-sibling binary tree \cite{CormenLeRi09}. The transformation of one such node with four children is shown in Figure \ref{fig:karytransformation}. 

A node with $k$ children is equivalent with respect to steals to a left-child right-sibling binary tree of height $k-1$. We show this fact for the setting in Figure \ref{fig:karytransformation}; the proof for the general setting is similar. For both trees in Figure \ref{fig:karytransformation}, the first steal takes subtree $T_4$ and leaves the tree containing subtrees $T_1,T_2$, and $T_3$. The second steal takes subtree $T_3$ and leaves the tree containing subtrees $T_1$ and $T_2$. Finally, the third steal takes subtree $T_2$ and leaves subtree $T_1$.

Consequently, we can use the same algorithm as in the case of binary trees to compute the maximum number of successful steals, as the next theorem shows. Recall from Section \ref{sec:notation} that $|T|$ denotes the size of the tree $T$.

\begin{figure}
\centering

\begin{tabular}{@{}c@{}}
\begin{tikzpicture}[scale=1]

  \path (0,2) coordinate (P1);
  \path (1,2) coordinate (P2);
  \path (2,2) coordinate (P3);
  \path (3,2) coordinate (P4);
  \path (1.5,3) coordinate (P5);
  \foreach \i in {5,...,5}
  {
    \fill (P\i) circle (2pt);
  }
  \draw (P1)--(P5) (P2)--(P5) (P3)--(P5) (P4)--(P5);
  \draw (0,1.7) circle(0.3) node {$T_1$}; 
  \draw (1,1.7) circle(0.3) node {$T_2$}; 
  \draw (2,1.7) circle(0.3) node {$T_3$}; 
  \draw (3,1.7) circle(0.3) node {$T_4$}; 

\end{tikzpicture}
\\[\abovecaptionskip]
\small (a) Tree whose root node has 4 children
\end{tabular}

\vspace{5mm}

\begin{tabular}{@{}c@{}}
\begin{tikzpicture}[scale=1]

  \path (1,2) coordinate (Q1);
  \path (2,2) coordinate (Q2);
  \path (1.5,3) coordinate (Q3);
  \path (0.5,1) coordinate (Q4);
  \path (1.5,1) coordinate (Q5);
  \path (0,0) coordinate (Q6);
  \path (1,0) coordinate (Q7);
  \foreach \i in {1,3,4}
  {
    \fill (Q\i) circle (2pt);
  }
  \draw (Q1)--(Q3) (Q2)--(Q3) (Q1)--(Q4) (Q1)--(Q5) (Q4)--(Q6) (Q4)--(Q7);
  \draw (0,-0.3) circle(0.3) node {$T_1$}; 
  \draw (1,-0.3) circle(0.3) node {$T_2$}; 
  \draw (1.5,0.7) circle(0.3) node {$T_3$}; 
  \draw (2,1.7) circle(0.3) node {$T_4$}; 

\end{tikzpicture}
\\[\abovecaptionskip]
\small (b) Tree resulting from transformation of root node
\end{tabular}

\caption{Transformation of root node in rooted tree} \label{fig:karytransformation}

\end{figure}

\begin{theorem}
Let $T$ be a rooted tree. There exists an algorithm that computes the potential $\Phi(T,n)$ in time $O(|T|n)$ and space $O(hkn)$, where $h$ denotes the height of $T$, and $k$ the maximum number of children of any node in $T$.
\end{theorem}

\begin{proof}
We transform $T$ into a binary tree according to the discussion leading up to this theorem, and apply the algorithm described in Corollary \ref{cor:binaryalgo}. Even though the transformed tree can be larger than the original one, it is no more than twice as large. Indeed, transforming a node with $k>2$ children introduces $k-2$ extra nodes in the resulting binary tree. 

The transformation takes time of order $|T|$, and the algorithm takes time of order $|T|n$. Since the transformed tree has height no more than $kn$, the algorithm takes space $O(hkn)$.
\end{proof}

\subsection*{Complete $k$-ary Trees}

As in our analysis of binary trees, an interesting special case is the case of complete $k$-ary trees, i.e., full $k$-ary trees in which all leaves have the same depth and every parent has $k$ children. Moreover, we determine the answer for almost complete $k$-ary trees. Recall from Section \ref{sec:notation} that for $k\geq 2,h\geq 0$, and $1\leq b\leq k-1$, $\id{ACT}(b,k,h)$ denotes the almost complete $k$-ary tree with $b\cdot k^h$ leaves.

\begin{theorem}
\label{thm:completekaryformula}
For $k\geq 2,h\geq 0$, and $1\leq b\leq k-1$, we have
\begin{align}
\Phi(\id{ACT}(b,k,h),n)={}  & \left((k-1)\binom{h}{1}+(k-1)^2\binom{h}{2}+\cdots+(k-1)^{n}\binom{h}{n}\right) \notag \\
                                   &+ (b-1)\left(\binom{h}{0}+(k-1)\binom{h}{1}+\cdots+(k-1)^{n-1}\binom{h}{n-1}\right) \notag \\
                                 ={} & \sum_{i=1}^n (k-1)^i\binom{h}{i} + (b-1)\sum_{i=0}^{n-1}(k-1)^i\binom{h}{i}. \notag 
\end{align} 
\end{theorem}

\begin{proof}
We first show that the formula is consistent even if we allow the tree $\id{ACT}(1,k,h+1)$ to be written as $\id{ACT}(k,k,h)$ for $h\geq 0$. Indeed, we have 
\begin{align}
\Phi(\id{ACT}(1,k,h+1),n) ={} & (k-1)\binom{h+1}{1}+(k-1)^2\binom{h+1}{2}+\cdots+(k-1)^{n}\binom{h+1}{n}  \notag \\
                                      ={} & (k-1)\left(\binom{h}{0}+\binom{h}{1}\right)+(k-1)^2\left(\binom{h}{1}+\binom{h}{2}\right) \notag \\
                                      & +\cdots+(k-1)^n\left(\binom{h}{n-1}+\binom{h}{n}\right) \notag \\
                                      ={} & \left((k-1)\binom{h}{1}+(k-1)^2\binom{h}{2}+\cdots+(k-1)^n\binom{h}{n}\right) \notag \\
                                      & + (k-1)\left(\binom{h}{0}+(k-1)\binom{h}{1}+\cdots+(k-1)^{n-1}\binom{h}{n-1}\right)  \notag \\
                                      ={} & \Phi(\id{ACT}(k,k,h),n), \notag
\end{align} 
where we used Pascal's identity (Equation \ref{eq:pascal}).

We proceed to prove the formula by induction on $b\cdot k^h$. The case $b\cdot k^h=1$ holds, since both the left-hand side and the right-hand side are zero. The case $n=0$ holds similarly. Consider the tree $\id{ACT}(b,k,h)$, where $b\cdot k^h>1$ and $2\leq b\leq k$. Using the consistency of the formula that we proved above, it is safe to represent any nontrivial tree in such form. Now, the recurrence in Theorem \ref{thm:worstcaserecurrence} yields
\[
\Phi(\id{ACT}(b,k,h),n) =1 + \max\{A,B\},
\]
where $A=\Phi(\id{ACT}(b-1,k,h),n)+\Phi(\id{ACT}(1,k,h),n-1)$ and $B=\Phi(\id{ACT}(b-1,k,h),n-1)+\Phi(\id{ACT}(1,k,h),n)$.

Using the induction hypothesis, we have 
\begin{align}
A &= \left(\sum_{i=1}^n(k-1)^i\binom{h}{i}+(b-2)\sum_{i=0}^{n-1}(k-1)^i\binom{h}{i}\right)+\sum_{i=1}^{n-1}(k-1)^i\binom{h}{i} \notag \\
  &= \left(\sum_{i=1}^n(k-1)^i\binom{h}{i}+(b-1)\sum_{i=0}^{n-1}(k-1)^i\binom{h}{i}\right)-1 \notag 
\end{align} 
and
\begin{align*}
B &= \left(\sum_{i=1}^{n-1}(k-1)^i\binom{h}{i}+(b-2)\sum_{i=0}^{n-2}(k-1)^i\binom{h}{i}\right)+\sum_{i=1}^{n}(k-1)^i\binom{h}{i} \notag \\
  &= A - (b-2)(k-1)^{n-1}\binom{h}{n-1} \notag \\
  &\leq A. 
\end{align*} 
Therefore, we have
\begin{align*}
\Phi(\id{ACT}(b,k,h),n) &= 1+A \notag \\
                                  &= \sum_{i=1}^n(k-1)^i\dbinom{h}{i}+(b-1)\sum_{i=0}^{n-1}(k-1)^i\dbinom{h}{i}, \notag 
\end{align*} 
as desired.
\end{proof}

It is worth noting that the exponential factor $(k-1)^i$ becomes huge when $k>2$. Instead, if we were to convert a node with $k>2$ children into binary nodes forming a balanced binary tree, we would only have to replace the height $h$ of the tree by $h\log k$. It should not be surprising that a balanced binary tree yields a much better upper bound than a node with multiple children. Indeed, as we have shown, a node with $k>2$ children is equivalent to a left-child right-sibling binary tree, which is highly unbalanced. When the tree is highly unbalanced, it is possible that a large number of steals are generated if the execution of smaller subtrees always finishes before that of larger ones. When the tree is balanced, on the other hand, the execution is more likely to reach a point where every processor owns a relatively large tree, and the execution of one such tree has to finish before the next steal can be performed.

We have established the upper bound on the number of steals in the configuration with one processor having an arbitrary rooted tree at the beginning. In the next section, we generalize one step further by allowing any number of processors to own work at the beginning.

\section{Rooted Tree and Multiple Starting Processors}
\label{sec:manystartingprocs}

In this section, we consider a generalization of Section \ref{sec:karytrees} where the work is not limited to one processor at the beginning, but rather can be spread out as well. We derive a formula for computing the potential function of a configuration based on the potential function of the individual trees. This leads to an algorithm that computes the upper bound for the configuration with trees $T_1,T_2,\ldots,T_P$ in time $O(P^3+P(|T_1|+|T_2|+\cdots+|T_P|))$. We then show that for complete $k$-ary trees, we only need to sort the trees in order to compute the maximum number of steals. Since we can convert any rooted tree into a binary tree as described in Section \ref{sec:karytrees}, it suffices throughout this section to analyze the case in which all trees are binary trees.

\subsection*{Formula}

Suppose that in our configuration, the $P$ processors start with trees $T_1,T_2,\ldots,T_P$. How might a sequence of steals proceed? As in our previous analysis of the case with one starting processor, we have an option of being greedy. We pick one tree---say $T_1$---and obtain as many steals as possible out of it using one processor. After we are done with $T_1$, we pick another tree---say $T_2$---and obtain as many steals as possible out of it using two processors. We proceed in this way until we pick the last tree---say $T_P$---and obtain as many steals out of it using all $P$ processors.

We make the following definition.

\begin{definition}
Let $T_1,T_2,\ldots,T_n$ be binary trees. Then $\Phi(T_1,T_2,\ldots,T_n)$ is the maximum number of steals that we can get from a configuration of $n$ processors that start with the trees $T_1,T_2,\ldots,T_n$.
\end{definition}

Note that we are overloading the potential function operator $\Phi$. Unlike the previous definition of $\Phi$, this definition does not explicitly include the number of processors, because this number is simply the number of trees included in the argument of $\Phi$. It follows from the definition that $\Phi(T_1,T_2,\ldots,T_n)=\Phi(T_{\sigma(1)},T_{\sigma(2)},\ldots,T_{\sigma(n)})$ for all permutations $\sigma$ of $1,2,\ldots,n$.

From the discussion leading up to the definition, we have  
\[\Phi(T_1,T_2,\ldots,T_P)\geq \Phi(T_{\sigma(1)},0)+\Phi(T_{\sigma(2)},1)+\cdots+\Phi(T_{\sigma(P)},P-1)\]
for any permutation $\sigma$ of $1,2,\ldots,P$. It immediately follows that 
\[\Phi(T_1,T_2,\ldots,T_P)\geq \max_{\sigma\in S_P}\left(\Phi(T_{\sigma(1)},0)+\Phi(T_{\sigma(2)},1)+\cdots+\Phi(T_{\sigma(P)},P-1)\right),\]
where $S_P$ denotes the symmetric group of order $P$, i.e., the set of all permutations of $1,2,\ldots,P$.

The next theorem shows that this inequality is in fact an equality.

\begin{theorem}
Let $T_1,T_2,\ldots,T_P$ be binary trees. We have
\[\Phi(T_1,T_2,\ldots,T_P)= \max_{\sigma\in S_P}\left(\Phi(T_{\sigma(1)},0)+\Phi(T_{\sigma(2)},1)+\cdots+\Phi(T_{\sigma(P)},P-1)\right).\]
\end{theorem}

\begin{proof}
We have already shown that the left-hand side is no less than the right-hand side, hence it only remains to show the reverse inequality.

Suppose that we are given any sequence of steals performed on $T_1,T_2,\ldots,T_P$ using the $P$ processors. Each steal is performed on a subtree of one of the trees $T_1,T_2,\ldots,T_P$. 

Assume without loss of generality that the last steal performed on a subtree of $T_1$ occurs before the last steal performed on a subtree of $T_2$, which occurs before the last steal performed on a subtree of $T_3$, and so on. That means that at any particular point in the stealing sequence, subtrees of $T_i$ occupy at most $i$ processors, for all $1\leq i\leq P$. (Subtrees of $T_i$ may have occupied a total of more than $i$ processors at different points in the stealing sequence, but that does not matter.) We can canonicalize the sequence of steals in such a way that all steals on subtrees of $T_1$ are performed first using one processor, then all steals on subtrees of $T_2$ are performed using two processors, and so on, until all steals on subtrees of $T_P$ are performed using $P$ processors. Therefore, in this case the total number of steals is no more than $\Phi(T_1,0)+\Phi(T_2,1)+\cdots+\Phi(T_P,P-1)$.

In general, let $\sigma$ be the permutation of $1,2,\ldots,P$ such that the last steal performed on a subtree of $T_{\sigma(1)}$ occurs before the last steal performed on a subtree of $T_{\sigma(2)}$, which occurs before the last steal performed on a subtree of $T_{\sigma(3)}$, and so on. Then we have \[\Phi(T_1,T_2,\ldots,T_P)\leq  \Phi(T_{\sigma(1)},0)+\Phi(T_{\sigma(2)},1)+\cdots+\Phi(T_{\sigma(P)},P-1).\] Therefore, \[\Phi(T_1,T_2,\ldots,T_P)\leq \max_{\sigma\in S_P}\left(\Phi(T_{\sigma(1)},0)+\Phi(T_{\sigma(2)},1)+\cdots+\Phi(T_{\sigma(P)},P-1)\right),\] which gives us the desired equality.
\end{proof}

\subsection*{Algorithm}

Now that we have a formula to compute $\Phi(T_1,T_2,\ldots,T_P)$, we again ask how fast we can compute it. Recall from Section \ref{sec:notation} that $|T|$ denotes the size of the tree $T$.

\begin{corollary}
\label{cor:algmanyprocessors}
There exists an algorithm that computes the potential $\Phi(T_1,T_2,\ldots,T_P)$ in time $O(P^3+P(|T_1|+|T_2|+\cdots+|T_P|))$.
\end{corollary}

\begin{proof}
The potentials $\Phi(T_i,j)$ can be precomputed by using dynamic programming in time $O(P(|T_1|+|T_2|+\cdots+|T_P|))$ using the algorithm in Corollary \ref{cor:binaryalgo}. It then remains to determine the maximum value of $\Phi(T_{\sigma(1)},0)+\Phi(T_{\sigma(2)},1)+\cdots+\Phi(T_{\sigma(P)},P-1)$ over all permutations $\sigma$ of $1,2,\ldots,P$. A brute-force solution that tries all possible permutations $\sigma$ of $1,2,\ldots,n$ takes time $O(P!)$. Nevertheless, our maximization problem is an instance of the assignment problem, which can be solved using the classical ``Hungarian method''. The algorithm by Tomizawa \cite{Tomizawa72} solves the assignment problem in time $O(P^3)$. Hence, the total running time is $O(P^3+P(|T_1|+|T_2|+\cdots+|T_P|))$.
\end{proof}

\subsection*{Complete Trees}

It is interesting to ask whether we can do better than the algorithm in Corollary \ref{cor:algmanyprocessors} in certain special cases. Again, we consider the case of complete trees. In this subsection we assume that the $P$ processors start with almost complete $k$-ary trees (defined in Section \ref{sec:notation}) for the same value of $k$. 

Suppose that at the beginning of the execution of work stealing, the processors start with the trees $\id{ACT}(b_1,k,h_1),\id{ACT}(b_2,k,h_2),\ldots,\id{ACT}(b_P,k,h_P)$, where $1\leq b_1,b_2,\ldots,b_P\leq k-1$. We may assume without loss of generality that $b_1\cdot k^{h_1}\leq b_2\cdot k^{h_2}\leq\cdots\leq b_P\cdot k^{h_P}$. Intuitively, in order to generate the maximum number of successful steals, one might want to allow larger trees more processors to work with, because larger trees can generate more steals than smaller trees. It turns out that in the case of complete trees, this intuition always works, as is shown in the following theorem.

\begin{theorem}
\label{thm:completekarymany}
Let $b_1\cdot k^{h_1}\leq b_2\cdot k^{h_2}\leq\cdots\leq b_P\cdot k^{h_P}$. We have
\begin{align*}
 \Phi(\id{ACT}(b_1,k,h_1),\id{ACT}(b_2,k,h_2),\ldots,\id{ACT}(b_P,k,h_P)) = \sum_{i=1}^P\Phi(\id{ACT}(b_i,k,h_i),i-1).
\end{align*}
\end{theorem}

\begin{proof}
We use the exchange argument: if the trees are not already ordered in increasing size, then there exist two consecutive positions $j$ and $j+1$ such that $b_j\cdot k^{h_j}>b_{j+1}\cdot k^{h_{j+1}}$. We show that we may exchange the positions of the two trees and increase the total potential in the process. Since we can always perform a finite number of exchanges to obtain the increasing order, and we know that the total potential increases with each exchange, we conclude that the maximum potential is obtained exactly when the trees are ordered in increasing size.

It only remains to show that any exchange of two trees $b_j\cdot k^{h_j}$ and $b_{j+1}\cdot k^{h_{j+1}}$ such that $b_j\cdot k^{h_j}>b_{j+1}\cdot k^{h_{j+1}}$ increases the potential. Denote the new potential after the exchange by $N$ and the old potential before the exchange by $O$. We would like to show that $N>O$. We have
\begin{align*}
N-O ={} & \left(\sum_{i=1}^j(k-1)^i\binom{h_j}{i}+(b_j-1)\sum_{i=0}^{j-1}(k-1)^i\binom{h_j}{i}\right)  \\
    &+    \left(\sum_{i=1}^{j-1}(k-1)^i\binom{h_{j+1}}{i}+(b_{j+1}-1)\sum_{i=0}^{j-2}(k-1)^i\binom{h_{j+1}}{i}\right) \\
    &-    \left(\sum_{i=1}^{j-1}(k-1)^i\binom{h_j}{i}+(b_j-1)\sum_{i=0}^{j-2}(k-1)^i\binom{h_j}{i}\right)  \\
    &-    \left(\sum_{i=1}^{j}(k-1)^i\binom{h_{j+1}}{i}+(b_{j+1}-1)\sum_{i=0}^{j-1}(k-1)^i\binom{h_{j+1}}{i}\right) \\
    ={} & (k-1)^j\binom{h_j}{j}+(b_j-1)(k-1)^{j-1}\binom{h_j}{j-1}  \\   
    &-    (k-1)^j\binom{h_{j+1}}{j}-(b_{j+1}-1)(k-1)^{j-1}\binom{h_{j+1}}{j-1}.
\end{align*} 

We consider two cases.

\underline{Case 1}: $h_j=h_{j+1}$ and $b_j>b_{j+1}$.

We have 
\begin{align*}
N-O ={} & (k-1)^j\binom{h_j}{j}+(b_j-1)(k-1)^{j-1}\binom{h_j}{j-1}  \\   
    &-    (k-1)^j\binom{h_j}{j}-(b_{j+1}-1)(k-1)^{j-1}\binom{h_j}{j-1} \\
    ={} & (b_j-1)(k-1)^{j-1}\binom{h_j}{j-1} - (b_{j+1}-1)(k-1)^{j-1}\binom{h_j}{j-1} \\
    ={} & ((b_j-1)-(b_{j+1}-1))(k-1)^{j-1}\binom{h_j}{j-1} \\
    ={} & (b_j-b_{j+1})(k-1)^{j-1}\binom{h_j}{j-1} \\
    >{} & 0,
\end{align*}

since $b_j>b_{j+1}$.

\underline{Case 2}: $h_j>h_{j+1}$.

We have 
\begin{align*}
N-O ={} & (k-1)^j\binom{h_j}{j}+(b_j-1)(k-1)^{j-1}\binom{h_j}{j-1}  \\ 
    &-    (k-1)^j\binom{h_{j+1}}{j}-(b_{j+1}-1)(k-1)^{j-1}\binom{h_{j+1}}{j-1} \\
    \geq{} & (k-1)^j\binom{h_j}{j} - (k-1)^j\binom{h_{j+1}}{j}-(b_{j+1}-1)(k-1)^{j-1}\binom{h_{j+1}}{j-1} \\
    \geq{} & (k-1)^j\binom{h_j}{j} - (k-1)^j\binom{h_{j+1}}{j}-(k-1)(k-1)^{j-1}\binom{h_{j+1}}{j-1} \\
    ={} & (k-1)^j\binom{h_j}{j} - (k-1)^j\left(\binom{h_{j+1}}{j}+\binom{h_{j+1}}{j-1}\right) \\
    ={} & (k-1)^j\binom{h_j}{j} - (k-1)^j\binom{h_{j+1}+1}{j} \\
    ={} & (k-1)^j\left(\binom{h_j}{j} -\binom{h_{j+1}+1}{j}\right) \\
    >{} & 0,
\end{align*}
since $h_j\geq h_{j+1}+1$.

In both cases, we have $N-O>0$, and hence any exchange increases the potential, as desired.
\end{proof}

It follows from Theorem \ref{thm:completekarymany} and Corollary \ref{cor:algmanyprocessors} that in the case of complete $k$-ary trees, one only needs to sort the trees with respect to their size in order to compute the potential. It follows that the running time of the algorithm is bounded by the running time of sorting, which is $O(P\lg P)$.

\section{Conclusion and Future Work}
\label{sec:conclusion}

In this paper, we have established tight upper bounds on the number of steals when the computation can be modeled by rooted trees. Here we suggest two possible directions for future work:

\begin{itemize}
\item This paper restricts the attention to computation trees. In general, however, computations need not be trees. How do the upper bounds generalize to the case where the computation can be modeled by directed acyclic graphs (DAG)?
\item Consider a model in which executing each task requires a constant amount of time independent of the size of the computation trees. In this model, if no processor steals from a particular processor for a long enough period, that processor will finish executing its own tasks and can no longer be stolen from. Can we achieve tighter bounds in this model?
\end{itemize}

\section{Note}
\label{sec:acknowledgements}

The final publication is available at Springer via \begin{center} \url{http://dx.doi.org/10.1007/s00224-015-9613-9}. \end{center}


\begin{thebibliography}{99}

\bibitem{AgrawalHeLe06} Kunal Agrawal, Yuxiong He, and Charles E. Leiserson, An empirical evaluation of work stealing with parallelism feedback. In \emph{26th IEEE International Conference on Distributed Computing Systems (ICDCS 2006)}, July 2006.

\bibitem{AroraBlPl98} Nimar S. Arora, Robert D. Blumofe, and C. Greg Plaxton. Thread scheduling for multiprogrammed multiprocessors. In \emph{Proceedings of the Tenth Annual ACM Symposium on Parallel Algorithms and Architectures (SPAA)}, pages 119--129, June 1998.

\bibitem{BlumofeLe99} Robert D. Blumofe and Charles E. Leiserson. Scheduling multithreaded computations by work stealing.  \emph{Journal of the ACM}, 46(5):720--748, 1999.

\bibitem{BurtonSl81} F. Warren Burton and M. Ronan Sleep. Executing functional programs on a virtual tree of processors. In \emph{Proceedings of the 1981 Conference on Functional Programming Languages and Computer Architecture}, pages 187--194, 1981.

\bibitem{CormenLeRi09} Thomas H. Cormen, Charles E. Leiserson, Ronald L. Rivest, and Clifford Stein, \emph{Introduction to Algorithms}.  The MIT Press, 3rd edition, 2009.

\bibitem{DinanLaSaKrNi09} James Dinan, D. Brian Larkins, P. Sadayappan, Sriram Krishnamoorthy, and Jarek Nieplocha. Scalable work stealing. In \emph{Proceedings of the Conference on High Performance Computing Networking, Storage and Analysis (SC)}, November 2009.

\bibitem{Halstead84} Robert H. Halstead, Jr. Implementation of Multilisp: Lisp on a multiprocessor. In \emph{Proceedings of the 1984 ACM Symposium on LISP and Functional Programming}, pages 9--17, 1984.

\bibitem{KarpZh93} Richard M. Karp and Yanjun Zhang. Randomized parallel algorithms for backtrack search and branch-and-bound computation. \emph{Journal of the ACM}, 40(3):765--789, July 1993.

\bibitem{Suksompong14} Warut Suksompong, Bounds on multithreaded computations by work stealing. Master's Thesis, Massachusetts Institute of Technology, 2014.

\bibitem{SuksompongLeSc16} Warut Suksompong, Charles E. Leiserson, and Tao B. Schardl. On the efficiency of localized work stealing. \emph{Information Processing Letters}, 116(2):100--106, 2016.

\bibitem{Tomizawa72} N. Tomizawa. On some techniques useful for solution of transportation network problems. \emph{Networks}, 1(2), 1972.

\end{thebibliography}



\end{document}